\documentclass{cccg17}
\usepackage{graphicx,amssymb,amsmath}

\usepackage{subfig}
\usepackage{enumitem}

\newcommand{\B}{\mathcal{B}}

\title{Power domination on triangular grids}

\author{Prosenjit Bose
\thanks{School of Computer Science, Carleton University, Ottawa ON, Canada. Research supported in part by NSERC. {\tt jit@scs.carleton.ca}}
        \and
        Claire Pennarun\thanks{LaBRI, Univ. Bordeaux, UMR 5800 {\tt  claire.pennarun@labri.fr}}
        \and 
        Sander Verdonschot\thanks{School of Computer Science, Carleton University, Ottawa ON, Canada. Research supported in part by NSERC. {\tt sander@cg.scs.carleton.ca}}}

\begin{document}
\thispagestyle{empty}
\maketitle

\begin{abstract}
The concept of \emph{power domination} emerged from the problem of monitoring electrical systems. Given a graph $G$ and a set $S \subseteq V(G)$, a set $M$ of monitored vertices is built as follows: at first, $M$ contains only the vertices of $S$ and their direct neighbors, and then each time a vertex in $M$ has exactly one neighbor not in $M$, this neighbor is added to $M$. 
The \emph{power domination number} of a graph $G$ is the minimum size of a set $S$ such that this process ends up with the set $M$ containing every vertex of $G$.
We here show that the power domination number of a triangular grid $T_k$ with hexagonal-shape border of length $k-1$ is exactly $\left\lceil \dfrac{k}{3} \right\rceil$.
\end{abstract}

\section{Introduction}

Power domination is a problem that arose from the context of monitoring electrical systems~\cite{mili_90,baldwin_93}, and was reformulated in graph terms by Haynes et al.~\cite{haynes_02}.

Given a graph $G$ and a set $S\subseteq V(G)$, we build a set $M$ as follows: at first, $M$ is the closed neighborhood of $S$, i.e. $M = N[S]$, and then iteratively a vertex $u$ is added to $M$ if $u$ is the only neighbor of a monitored vertex $v$ that is not in $M$ (we say that $v$ \emph{propagates} to $u$).
At the end of the process, we say that $M$ is the set of vertices \emph{monitored} by $S$.
We say that $G$ is \emph{monitored} when all its vertices are monitored. 
The set $S$ is a \emph{power dominating set} of $G$ if $M = V(G)$, and the minimum cardinality of such a set is the \emph{power domination number} of $G$, denoted by $\gamma_P(G)$.

Power domination has been particularly well studied on regular grids and their generalizations: the exact power domination number has been determined for the square grid~\cite{dorfling_06} and other products of paths~\cite{dorbec_08}, for the hexagonal grid~\cite{ferrero_11}, as well as for cylinders and tori~\cite{barrera_11}. 
These results are particularly interesting in comparison with the ones on the same classes for (classical) domination: for example, the problem of finding the domination number of grid graphs $P_n \times P_m$ was a difficult problem which was solved only recently~\cite{goncalves_11}.
They also rely heavily on propagation: it is generally sufficient to monitor (with adjacency alone) a small portion of the graph in order to propagate to the whole graph.

We here continue the study of power domination in grid-like graphs by focusing on triangular grids with hexagonal-shaped border.

\bigskip
A \emph{triangular grid} $T_k$ has vertex set $V(T_k) = \{(x,y,z) \mid x,y,z \in [0..2k-2], x-y+z = k-1 \}$. Two vertices $(x,y,z)$ and $(x',y',z')$ are adjacent if and only if $|x'-x|+|y'-y| + |z'-z| = 2$. The graph $T_k$ has a regular hexagonal shape, and $k$ is the number of vertices on each edge of the hexagon. Figure~\ref{fig:def_grid} shows the two triangular grids $T_2$ and $T_3$. Note that $T_k$ appears as a subgraph of $T_{k+1}$ (where $(1,1,1)$ has been added to the coordinates of each vertex in $T_k$).

We prove the following theorem:
 \begin{theorem} \label{th:grid}
 For $k \in \mathbb{N}^*$, $\gamma_P(T_k) = \left\lceil  \dfrac{k}{3} \right\rceil$.
 \end{theorem}

\begin{figure}[ht]
\centering
\includegraphics[scale=0.5]{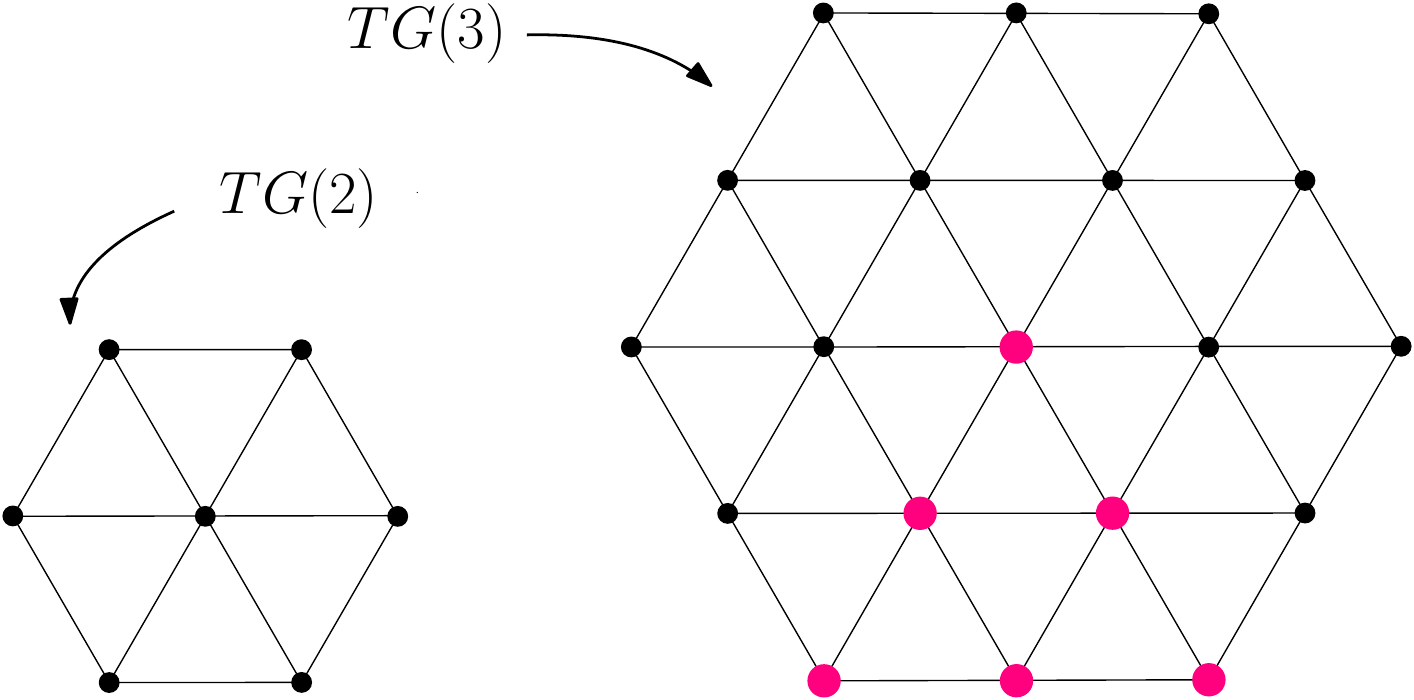}
\caption{The graphs $T_2$ and $T_3$, along with the coordinates of the vertices.}
\label{fig:def_grid}
\end{figure}

An inner vertex $v \in V(T_k)$ with coordinates $(x,y,z)$ has 6 neighbors with the following coordinates: $(x,y+1,z+1)$, $(x-1,y,z+1)$, $(x-1,y-1,z)$ , $(x,y-1,z-1)$ , $(x+1,y,z-1)$ and $(x+1,y+1,z)$ (see Figure~\ref{fig:grid_coordinates}). 
The coordinates of a vertex $v$ are denoted by $(v_1,v_2,v_3)$.
The \emph{line} $l_{v_j=i}$ is the set of vertices $\{(v_1,v_2,v_3) \mid v_j = i \}$ (see Figure~\ref{fig:line}).

\begin{figure}[ht]
\centering
\subfloat[]{
\includegraphics[scale=0.6]{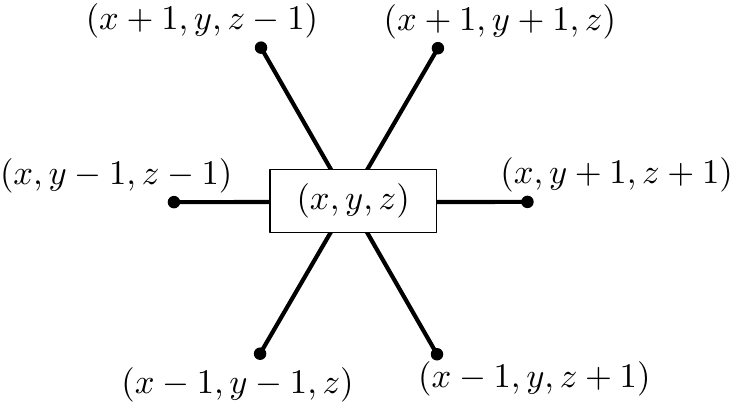}
\label{fig:grid_coordinates}}
\subfloat[]{
\includegraphics[scale=0.35]{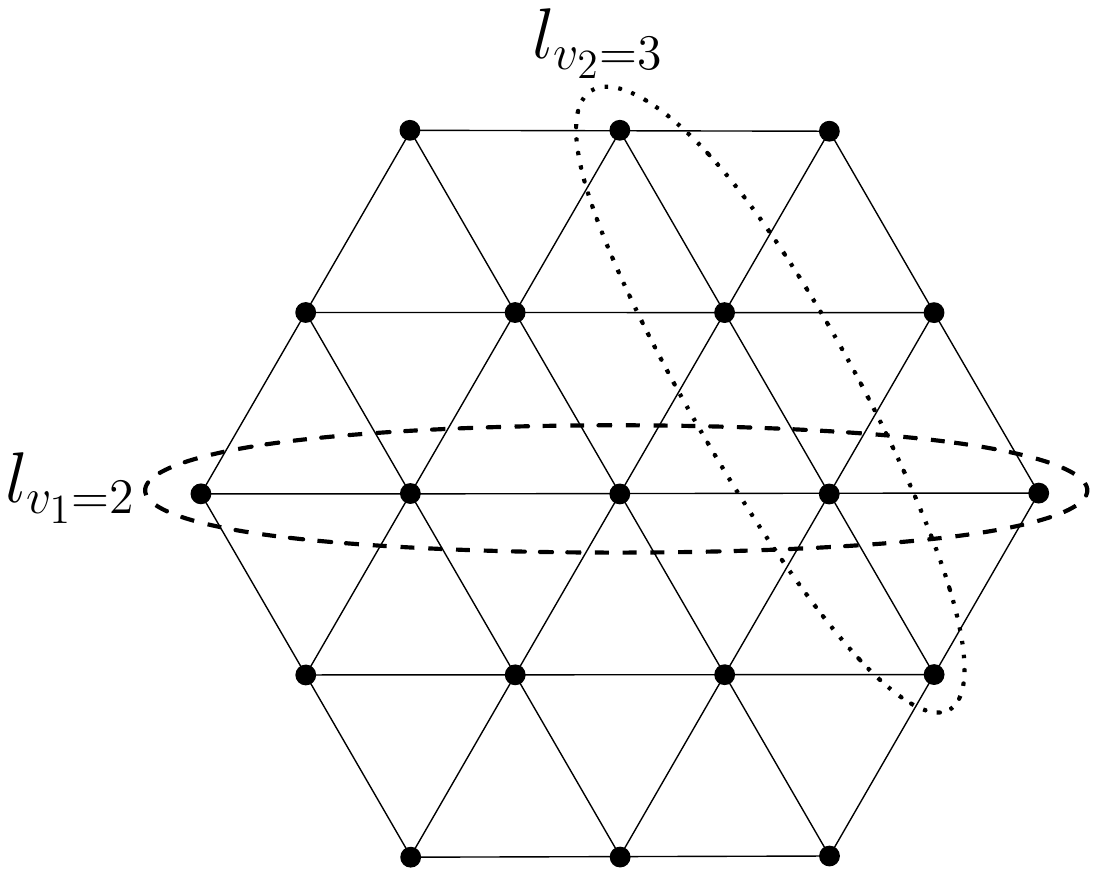}
\label{fig:line}}
\caption{(a) The coordinates of the neighbors around an inner vertex $v=(x,y,z)$. (b) The lines $l_{v_1=2}$ and $l_{v_2=3}$ in $T_3$.}
\end{figure}

One interesting property of the triangular grids is that if an equilateral triangle having one side of the hexagonal border as base is monitored, then the border allows the propagation until the whole graph is monitored. For example, it suffices to monitor the set $\mathcal{T} = \{v=(v_1, v_2,v_3) \in V(G) \mid 0 \leq v_1,v_2 \leq k-1, k-1 \leq v_3 \leq 2k-2\}$ to monitor $T_k$ (see Figure~\ref{fig:prop_hexa}).

\begin{figure}[ht]
\centering
\includegraphics[scale=0.5]{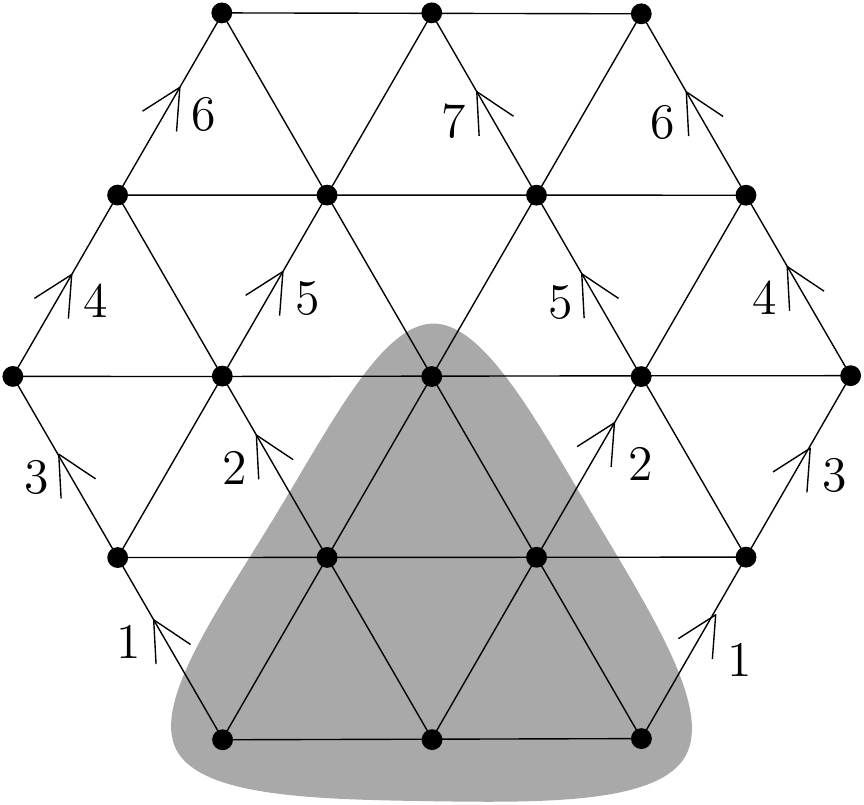}
\caption{The propagation steps to monitor $T_3$ once the set $\mathcal{T}$ (in the gray area) is monitored. Propagation steps indexed by the same number can be done in parallel.}
\label{fig:prop_hexa}
\end{figure}

We assume throughout the section that $k \geq 4$: observe that if $k \leq 3$, then $\gamma_P(T_k) = 1$, with $S = \{(k-2,k-2,k-1)\}$ (for $k=2,3$).
 
\section{Upper bound}

We begin by giving a construction for the upper bound:

\begin{lemma} \label{lem:upper}
For $k \in \mathbb{N}^*$, $\gamma_P(T_k) \leq \left\lceil  \dfrac{k}{3} \right\rceil$.
\end{lemma}

\begin{proof}
Let $i = \left\lfloor \frac{k}{3} \right\rfloor$, and $d = k-i-1$ if $k \equiv 0,1 \mod 3, d= k-i-2$ otherwise.
Let $S'$ be the following set of vertices (see Figure~\ref{fig:upper_bound}): $S' = \{ (1+3 \ell, d + \ell, k+d-2-2\ell), 0 \leq \ell \leq i-1\}$. In other words, $S'$ contains the vertex $v=(1,d,k+d-2)$ and vertices whose coordinates are obtained by adding $(3,1,-2)$ up to $i-1$ times to the coordinates of $v$.
If $k \not \equiv 0 \mod 3$, $S = S' \cup \{(k-1,k-1,k-1)\}$. Otherwise, $S=S'$.
Then we have, depending on the value of $k$ modulo 3:
\begin{itemize}
\item $k = 3i$: $|S| = i = \left\lceil \frac{3i}{3} \right\rceil$.
\item $k = 3i+1$: $|S| = i+1 = \left\lceil \frac{3i+1}{3} \right\rceil$.
\item $k = 3i+2$: $|S| = i+1 = \left\lceil \frac{3i+2}{3} \right\rceil$.
\end{itemize}
In each case, $S$ is a set with cardinality $\left\lceil \frac{k}{3} \right\rceil$, and $S$ progressively power dominates the 
%set $\mathcal{T}$ (see Figure~\ref{fig:upper_bound}), and thus the 
whole triangular grid $T_k$. 
\end{proof}

\begin{figure}[ht]
\centering
\includegraphics[scale=0.4]{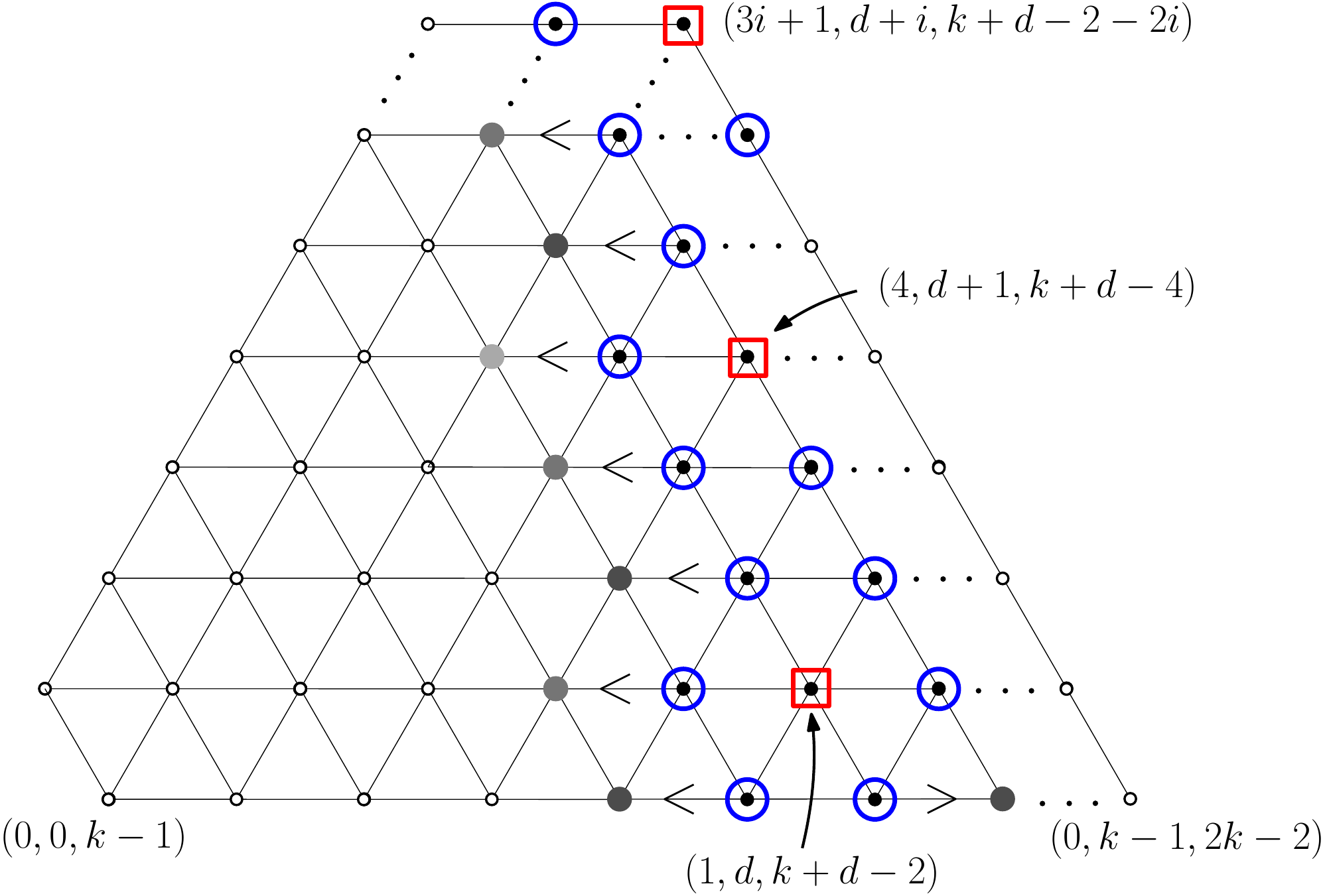}
\caption{Construction and propagation of the set $S'$: $d = k-i-1$ if $k \equiv 0,1 \!\mod 3$, $d= k-i-2$ if $k \equiv 2 \mod 3$. Red square-framed vertices are in $S'$, blue circle-framed vertices are in $N[S']$. Dark gray vertices are monitored in the first propagation round, gray ones in the second round, and the light gray one in the third round. 
Observe how the pattern of monitored vertices repeats.}
\label{fig:upper_bound}
\end{figure}

\section{Lower bound}

Let $A \subset V(T_k)$ be a set of vertices of the graph. 
We define the \emph{border} $\B_A \subseteq A$ of $A$ as follows: $\B_A = \{v \in A, N(v) \setminus A \neq \emptyset\}$.
Let $A_{v_j=i}$ denote the set of vertices of $A$ in a given line $l_{v_j=i}$. 
We define the \emph{$j$-shifted set} $A'= A^{(j)}$ of $A$ as follows (see Figure~\ref{fig:shift}): 
$|A'| = |A|$, and for each line $l_{v_j=i}$, 
$A'$ contains the $|A_{v_j=i}|$ vertices 
with smallest coordinates $v_{j+1}$ (for example, the 1-shifted set of $A$ contains only left-most vertices on each horizontal line). More formally,
\[A'_{v_j=i} = \{(v_1,v_2,v_3) \mid v_j=i, v_{j+1} = \ell + \alpha, 0 \leq \ell < |A_{v_j=i}| \},\quad \]
with $\alpha = 0$ if $0 \leq i \leq k-1$, and $\alpha = i - (k-1)$ if $k \leq i \leq 2k-2$.

%\[A'_{v_j=i} = \{(v_1,v_2,v_3) \mid v_j=i, v_{j+1} = \begin{cases} \ell & 0 \leq i \leq k-1 \\ \ell + i - (k-1) & k \leq i \leq 2k-2\end{cases}, 0 \leq \ell < |A_{v_j=i}| \} \,,\quad \]

\begin{figure}[ht]
\centering
\includegraphics[scale=0.34, page=1]{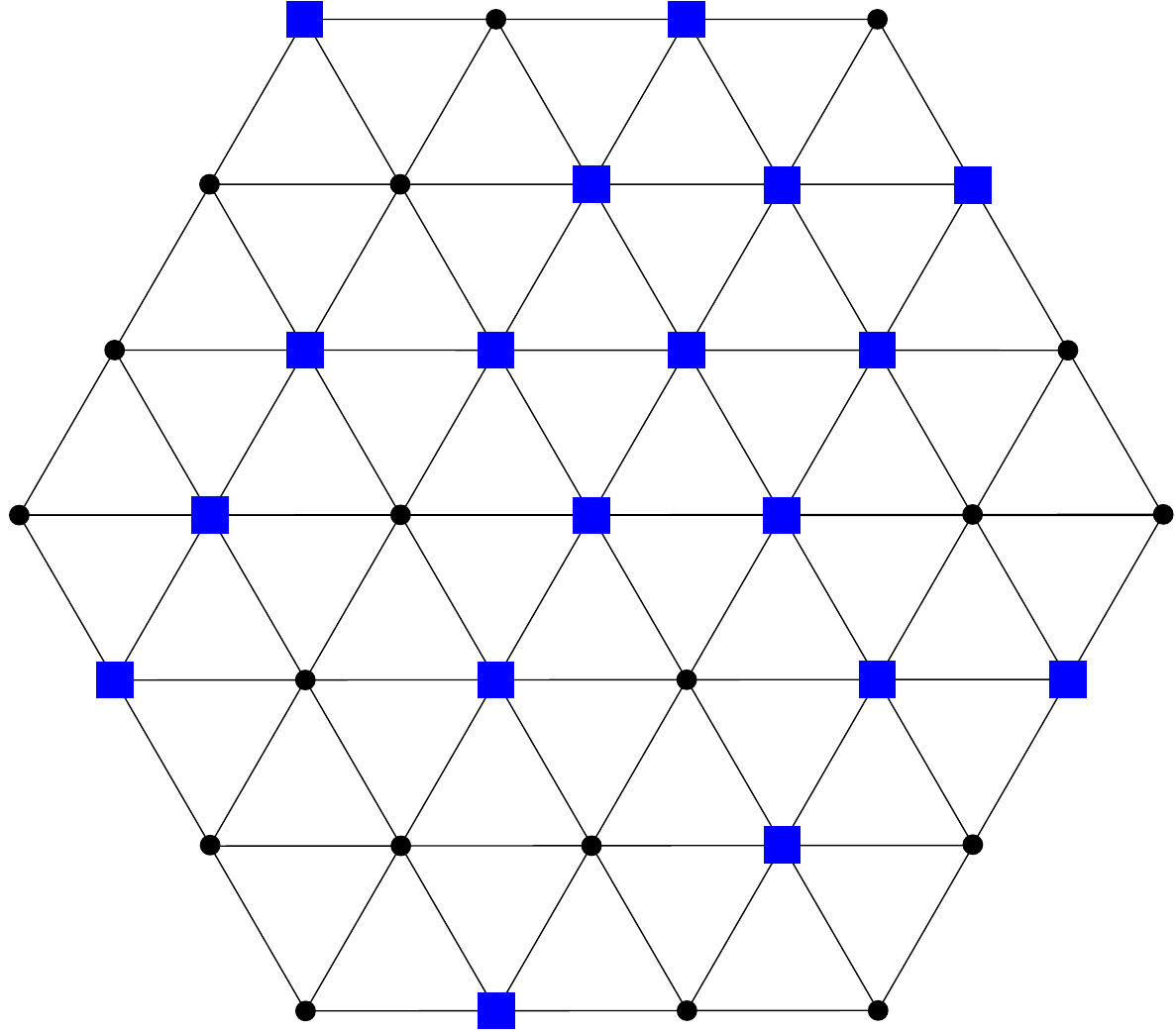}
\includegraphics[scale=0.34, page=2]{shift_power_dom.pdf}
\caption{(Left) Blue-square vertices are in the set $A$. (Right) Blue-square vertices are in the $1$-shifted set $A'$ of $A$: the left-most vertices of each line $l_{v_1 =i}$ are in $A'$.}
\label{fig:shift}
\end{figure}

\begin{lemma} \label{lem:border} Let $A'$ be the $j$-shifted set of $A$. Then $|\B_{A'}| \leq |\B_A|$. \end{lemma}

\begin{proof}
In this proof, since $j$ is fixed, we simplify the notation $l_{v_j=i}$ into $l_i$.
Let $a_i$ be the number of vertices in $A$ (and in $A'$) in line $l_i$ and $b_i$ (resp. $b'_i$) be the number of vertices in $\B_A$ (resp. $\B_{A'}$) in line $l_i$. 
We show that $b_i \geq b'_i$ for every line $l_i$, $0 \leq i \leq 2k-2$.
We consider three cases depending on the value of $i$ (when $0 \leq i < k-1$, when $i=k-1$ and when $k \leq i \leq 2k-2$):

\begin{itemize}
\item $0 \leq i < k -1$: we thus have 
$|l_{i+1}| = |l_i| + 1$ and $|l_i| = |l_{i-1}| + 1$.
Let us consider vertices in line $l_i$ which are in $A$ but not in the border of $A$: there are $a_i-b_i$ such vertices. By definition, we have $a_i - b_i \leq a_i$. Their neighbors (if they exist) in $l_{i-1}$ and $l_{i+1}$ are in $A$. We have thus both $a_i-b_i \leq a_{i+1}-1$, and $a_i-b_i \leq a_{i-1}$. Hence $a_i - b_i \leq \min\{a_{i+1}-1, a_{i-1}, a_i\}$ for $1 \leq i < k-1$ (for $i=0$, we have $a_i - b_i \leq \min\{a_{i+1}-1, a_i\}$).
We can apply the same reasoning to the vertices that are in $A'$ but not in the border of $A'$: since the vertices of $A'$ are consecutive on lines $l_{i-1}$, $l_i$ and $l_{i+1}$, we get that $a_i - b'_i = \min\{a_{i+1}-1, a_{i-1}, a_i\}$ (for $i=0$, we have $a_i - b'_i = \min\{a_{i+1}-1, a_i\}$). Note that the inequalities we get for $A$ turn into equalities on $A'$.
Then $a_i - b_i \leq a_i - b'_i$, and thus $b_i \geq b'_i$.

\item We have a similar proof when $k-1 < i \leq 2k-2$, for which we have $|l_{i+1}| = |l_i|-1$ and $|l_i| = |l_{i-1}|-1$: in that case, we get $a_i - b'_i = \min\{a_{i-1}-1,a_{i+1},a_i\} \geq a_i - b_i$.

\item $i=k-1$: we thus have $|l_{i+1}| = |l_{i-1}| = |l_i| +1$.
As for the previous case, first consider vertices which are in $A$ 
but not in the border of $A$: by definition $a_i - b_i \leq a_i$, 
and we have $a_{i+1} \geq a_i - b_i$ and $a_{i-1} \geq a_i - b_i$. 
Thus $a_i - b_i \leq \min\{a_{i+1}, a_{i-1}, a_i\}$.
Similarly, we get that $a_i - b'_i = \min\{a_{i+1}, a_{i-1}, a_i\}$.
Thus $a_i - b_i \leq a_i - b'_i$, and so $b_i \geq b'_i$. %\qedhere
\end{itemize}
\end{proof}

We define the \emph{shifting process} of a set $A \subset V(T_k)$ as the following iterative process: $A_{\ell+1}=((A_\ell^{(1)})^{(2)})^{(3)}$, with $A_0=A$. In other words, we successively apply 1-shift, 2-shift and 3-shift to the set $A$ until a fixed point $A_{\ell^*}$ is reached.
We show that this fixed point exists and that the vertices of the resulting set form a particular shape:

\begin{lemma} %\label{lem:stops}
\begin{enumerate}[label=\upshape(\roman*),ref=\thelemma(\roman*)]
\item\label{lem:stops} This shifting process stops, i.e. there exists $\ell^*$ such that $A_{\ell^*+1} = A_{\ell^*}$.
\item\label{lem:staircase} Let $A^* = A_{\ell^*}$. If $v= (x,y,z) \in A^*$, then all vertices $v'= (x',y',z')$ with $y' \leq y$ and $z' \leq z$ are also in $A^*$ (see Figure~\ref{fig:staircase1}). 
\end{enumerate}

\end{lemma}

\begin{figure}[ht]
\centering
\includegraphics[scale=0.6]{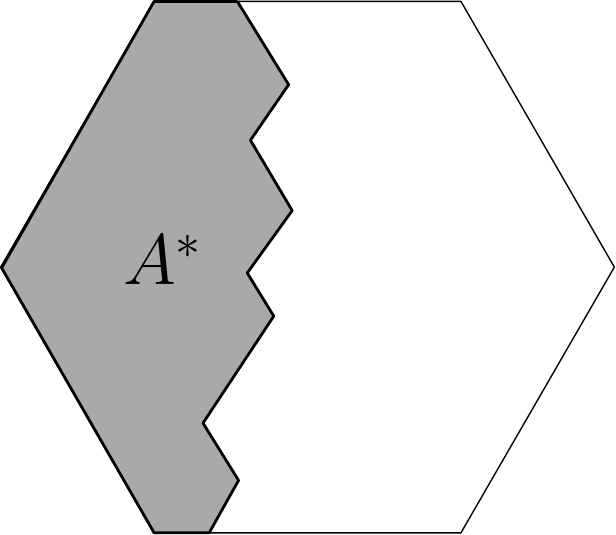}
\caption{The set $A^*$ has a staircase shape.}
\label{fig:staircase1}
\end{figure}

\begin{proof}

(i) We define the \emph{weight} in $A$ of a vertex as follows: $w_A(v) = v_1 + 2v_2 + 2v_3$ if $v \in A$, $w_A(v) = 0$ otherwise.
 Similarly, the weight of a set $S$ relatively to $A$ is $w_A(S) = \sum_{v \in S} w_A(v)$. For simplicity, we denote by $w_A$ the \emph{global weight} of the set $A$: $w_A = w_A(T_k)$.
 
Let $A'$ be the $j$-shifted set of $A$. We show that if $A' \neq A$, then $w_{A'} < w_A$.

Recall that for every vertex $v=(v_1,v_2,v_3)$ of $T_k$, $v_1 -v_2 + v_3 =k-1$.
We first show that if $v$ and $v'$ are two vertices with $v_j(v') = v_j(v)$ and $v_{j+1}(v') < v_{j+1}(v)$, then $w(v') < w(v)$:
\begin{itemize}
\item $j=1$: $v_1(v') = v_1(v)$ and $v_2(v') < v_2(v)$, so $v_3(v') = k-1 -v_1(v') + v_2(v') = k-1 - v_1(v) + v_2(v') < v_3(v)$. Thus $w(v') < w(v)$.
\item $j=2$: $v_2(v') = v_2(v)$ and $v_3(v') < v_3(v)$. 
Since $v_1(v)-v_2(v)+v_3(v) = v_1(v')-v_2(v')+v_3(v')$, we get $v_1(v) + v_3(v) = v_1(v')+v_3(v')$. Thus $w(v) -w(v') = v_1(v) + 2v_2(v) + 2v_3(v) - v_1(v') - 2v_2(v') - 2v_3(v') = v_3(v) - v_3(v')$. So $w(v') < w(v)$.
\item $j=3$: $v_3(v') = v_3(v)$ and $v_1(v') < v_1(v)$, so $v_2(v') = v_1(v') + v_3(v') -k+1 = v_1(v') + v_3(v) -k+1 \leq v_2(v)$. Thus $w(v') < w(v)$.
\end{itemize}
By definition on a $j$-shifted set, 
for each line $l_{v_j=i}$, 
\[w_{A'}(l_{v_j=i}) - w_A(l_{v_j=i}) = \sum_{v' \in A' \setminus A} w(v') - \sum_{v \in A \setminus A'} w(v)  \,,\]
and either $A_{v_j=i} = A'_{v_j=i}$, and this sums to 0, or $A_{v_j=i} \neq A'_{v_j=i}$, and it is strictly negative. Therefore  $A' \neq A$ implies $w_{A'} < w_A$.
Since the global weight of any set is always positive, this directly concludes the proof of item (i).

(ii) Let $v = (v_1,v_2,v_3)$ be a vertex in $A^*$. The vertices $u_1=(v_1+1,v_2,v_3-1)$, $u_2=(v_1,v_2-1,v_3-1)$ and $u_3=(v_1-1,v_2-1,v_3)$ (i.e. the north-west, west and south-west neighbors of $v$) are also in $A^*$: otherwise, we could again shift the set $A^*$ and get the set $A^*-\{v\}+\{u_i\}$, which has less weight than $A^*$, a contradiction.
Since this is true for every vertex of $A^*$, the proposition holds.
\end{proof}

We can now prove the lower bound:

\begin{lemma}\label{lem:lower}
For $k \in \mathbb{N}^*$, 
$\gamma_P(T_k) \geq \dfrac{2k-1}{6}$.
\end{lemma}

\begin{proof}
Let $S$ be a power dominating set of $T_k$. If $|S| > \frac{k}{3}$, then the result holds. Thus we assume $|S| \leq \left\lceil \frac{k}{3} \right\rceil$.
In power domination, propagation from a set $S$ is done by rounds. We decide of an arbitrary order on the vertices monitored by $S$ during each round. This defines a (non-unique) total order $m_1,\ldots, m_{|V(G) \setminus N[S]|}$ on the vertices of $V(G) \setminus N[S]$. We then define the set $M[t]$ as follows: $M[0]=N[S]$, and $M[t+1] = M[t] \cup \{m_{t+1}\}$.

The key idea of this proof is to consider the size of the sets $\B_{M[t]}$, to bound it and to deduce a bound on $|S|$. It is a classical way to prove lower bounds for power domination in regular lattices (see for example the lower bound proof on strong products~\cite{dorbec_08}). However, on the contrary to what happens in other cases, the size of the sets $\B_{M[t]}$ is not globally bounded from below: at the end of the propagation, no vertices belong to the border of the monitored set. We thus ``stop'' the propagation in the middle of the process and reason from there.

\medskip
\textbf{Claim 1.} For any $0 \leq i \leq |V(G) \setminus N[S]|$, we have $|\B_{M[i]}| \leq 6|S|$.

{\it Proof.}
We prove it by induction on $i$: $|\B_{M[0]}| = |\B_{N[S]}| \leq 6|S|$ by definition. If the vertex $m_{i+1}$ becomes monitored by propagation from a vertex $v$ in $\B_{M[i]}$, then $v$ is not in $\B_{M[i+1]}$, and at most one vertex ($m_{i+1}$) is added to $\B_{M[i+1]}$. Thus $|\B_{M[i+1]}| \leq |\B_{M[i]}|$. Using the induction hypothesis, we conclude that $|\B_{M[i+1]}| \leq 6|S|$.
{\footnotesize ($\square$) }

\medskip
Let $M$ be the set $M[t]$ containing $\frac{|V(T_k|}{2}$ vertices 
(as soon as $k \geq 3$, we get $\frac{|V(T_k|}{2} = \frac{3k^2 - 3k +1}{2} \geq \frac{7(k+1)}{3} \geq 7|S| \geq |M[0]|$, and so $M$ exists), and let $M^*$ be the set defined from $M$ by Lemma~\ref{lem:stops}. 

\smallskip
\textbf{Claim 2.} We have $2k-1 \leq |\B_{M^*}|$.

{\it Proof.}
We now prove that for every index $0 \leq i \leq 2k-2$, the line $l_{v_1=i}$ contains at least one vertex of $\B_{M^*}$.

Suppose there exists an index $0 \leq i \leq 2k-2$ such that all vertices of the line $l_{v_1=i}$ are in $M^*$. 
If $0\leq i \leq k-1$, then the vertex $w = (i,k+i-1,2k-2)$ (i.e. the right-most vertex of the line $l_{v_1=i}$) is in $M^*$, and so by Lemma~\ref{lem:staircase}, all vertices of the set $\{(v_1,v_2,v_3) \mid v_2 \leq k+i-1\}$ are also in $M^*$ (see Figure~\ref{fig:staircase}a). Since $k+i-1 > k-1$, then strictly more than half of the vertices of $T_k$ are in $M^*$, and so $M^*$ has strictly more than the required number of vertices, a contradiction.
Similarly, if $k-1 < i \leq 2k-2$: the vertex $w = (i,2k-2,3k-3-i)$ (i.e. the right-most vertex of the line $l_{v_1=i}$) is in $M^*$, and thus by Lemma~\ref{lem:staircase}, all vertices of the set $\{(v_1,v_2,v_3) \mid v_3 \leq 3k-3-i \}$ are also in $M^*$. Since $3k-3-i > k-1$, then strictly more than half of the vertices of $T_k$ are in $M^*$, a contradiction.
Thus every line $l_{v_1=i}$ contains at least one vertex not in $M^*$. 

Suppose now that one of the lines $l_{v_1=i}$ contains no vertex of $M^*$.
If $0 \leq i \leq k-1$ (see Figure~\ref{fig:staircase}b), then the vertex $w=(i,0,k-1-i)$ (i.e. the left-most vertex of the line $l_{v_1=i}$) is not in $M^*$. By the contrapositive of Lemma~\ref{lem:staircase}, the line $l_{v_3=k-1-i}$ also contains no vertices of $M^*$, and so all vertices of $M^*$ are included in the set $\{(v_1,v_2,v_3) \mid v_3 < k-1-i\}$ (they are all on the left and above line $l_{v_3=k-1-i}$). Thus $M^*$ contains strictly less than the half of the vertices of $T_k$, a contradiction.
Similarly, if $k-1 < i \leq 2k-2$, then the vertex $w = (i,i-k+1,0)$ is not in $M^*$. By the contrapositive of Lemma~\ref{lem:staircase}, the line $l_{v_2=i-k+1}$ also contains no vertices of $M^*$, and so all vertices of $M^*$ are included in the set $\{(v_1,v_2,v_3) \mid v_2 < i-k+1\}$ (they are all on the left and below line $l_{v_2=i-k+1}$). Since in that case $i-k+1 < k-1$, then again, $|M^*| = |M| < \frac{|V(T_k)|}{2}$ vertices, a contradiction.

\begin{figure}[ht]
\centering
\includegraphics[scale=0.6]{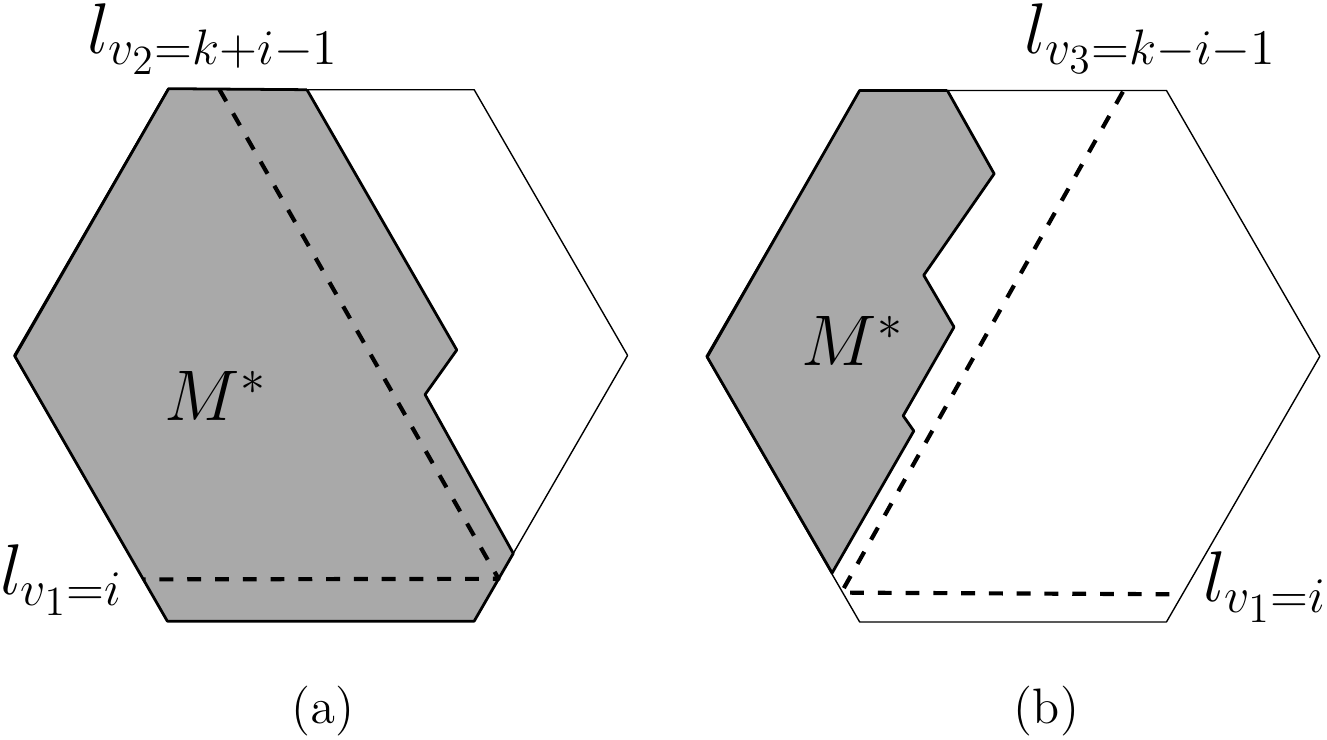}
\caption{(a) If all vertices of a line $l_{v_1=i}$ are in $M^*$ ($0 \leq i \leq k-1$), then all vertices of $T_k$ with $v_2 \leq k+i-1$ are also in $M^*$. (b) If the line $l_{v_1=i}$ contains no vertices of $M^*$ ($1 \leq i \leq k-1$), then all vertices of $M$ are above and left of line $l_{v_3=k-i-1}$.}
\label{fig:staircase}
\end{figure}

We thus get that each line $l_{v_1=i}$ contains at least one vertex of $M^*$ and not all its vertices are in $M^*$. Thus each line contains at least one vertex of $\B_{M^*}$, and so $2k-1 \leq |\B_{M^*}|$. {\footnotesize ($\square$) }

\smallskip
By Lemma~\ref{lem:border}, $|\B_{M^*}| \leq |\B_M|$, hence $2k-1 \leq |\B_M|$. Using Claim~1, we get $2k-1 \leq |\B_M| \leq 6|S|$, and so $|S| \geq \dfrac{2k-1}{6}$, which concludes the proof. \end{proof}

\medskip
We know that $\gamma_P(T_k)$ is an integer. Since there is no integer between $\frac{2k-1}{6} = \frac{k}{3} - \frac{1}{6}$ and $\left\lceil \frac{k}{3} \right\rceil$, then Lemma~\ref{lem:lower} directly implies that $\left\lceil \frac{k}{3} \right\rceil \leq \gamma_P(T_k)$.

This then gives our global result: 
\[\gamma_P(T_k) = \left\lceil \frac{k}{3} \right\rceil \quad \,,\]
concluding the proof of Theorem~\ref{th:grid}.

\section{Discussion}  
We carried on with the study of power domination in regular lattices, and examined the value of $\gamma_P(G)$ when $G$ is a triangular grid with hexagonal-shaped border. We showed that in that case, $\gamma_P(G) = \left\lceil \frac{k}{3} \right\rceil$.

The process of propagation in power domination led to the development of the concept of propagation radius, i.e. the number of propagation steps necessary in order to monitor the whole graph~\cite{dorbec_14}. It would be interesting to study the propagation radius of our constructions (in particular in the case of triangular grids) and to try and find a power dominating set minimizing this radius.

It seems that the border plays an important role in the propagation when the grid has an hexagonal shape, and so the next step in the understanding of power domination in triangular grids would be to look into grids with non-hexagonal shape. For example, what is the power domination number of a triangular grid with triangular border? 

Finally, the relation of our results with the ones presented for hexagonal grids by Ferrero et al.~\cite{ferrero_11} has to be noted: they show (with techniques different from the ones used in this paper) that $\gamma_P(H_n) =  \left\lceil \frac{2n}{3} \right\rceil$, where $n$ is the dimension of the hexagonal grid $H_n$, and so $\gamma_P(H_n) = \gamma_P(T_{2n})$.
Moreover, it is interesting to remark that $H_n$ is an induced subgraph of $T_{2n}$.
We already know~\cite{dorbec_16} that in general, the power domination number of an induced subgraph can be either smaller or arbitrarily large compared to the power domination number of the whole graph.
It would then be very interesting to investigate further under which conditions induced subgraphs have the same power dominating number as the whole graph.

%---------------------------- Bibliography -------------------------------

% Please add the contents of the .bbl file that you generate,  or add bibitem entries manually if you like.
% The entries should be in alphabetical order
\small
\bibliographystyle{abbrv}

\end{document}